\documentclass[conference]{IEEEtran}
\usepackage{cite}
\usepackage{amsthm,amsmath,amssymb,amsfonts,mathtools}
\usepackage{enumerate}
\theoremstyle{remark}
\newtheorem{proposition}{Proposition}
\newtheorem{remark}{Remark}
\newtheorem{definition}{Definition}
\newtheorem{theorem}{Theorem}
\newtheorem{example}{Example}

\title{Scalar Solvability of Network Computation Problems and Representable Matroids}
\begin{document}

\author{
\IEEEauthorblockN{Anindya Gupta and B. Sundar Rajan}
\IEEEauthorblockA{Department of Electrical Communication Engineering, Indian Institute of Science, Bengaluru 560012, KA, India\\Email:~\{anindya.g, bsrajan\}@ece.iisc.ernet.in}
}

\maketitle
\begin{abstract}
We consider the following \textit{network computation problem}. In an acyclic network, there are multiple source nodes, each generating multiple messages, and there are multiple sink nodes, each demanding a function of the source messages. The network coding problem corresponds to the case in which every demand function is equal to some source message, i.e., each sink demands some source message. Connections between network coding problems and matroids have been well studied. In this work, we establish a relation between network computation problems and representable matroids. 
We show that a network computation problem in which the sinks demand linear functions of source messages admits a scalar linear solution if and only if it is matroidal with respect to a representable matroid whose representation fulfills certain constraints dictated by the network computation problem. Next, we obtain a connection between network computation problems and functional dependency relations (FD-relations) and show that FD-relations can be used to characterize network computation problem with arbitrary (not necessarily linear) function demands as well as nonlinear network codes.
\end{abstract} 

\section{Introduction}
Conventional communication networks, like the Internet, ensure transfer of information generated at some nodes to others. It is known that network coding affords throughput gain over routing in such networks (see \cite{Ahls,Yeung,Ho} and references therein), and given a network and the demanded source messages at each sink, the network coding problem is to design a network code that maximizes the rate of information transfer from the source nodes to the sinks. But in some networks, like a sensor networks for environmental monitoring, nodes may be interested not in the messages generated by some other nodes but in one or more functions of these messages. Designing a network code that maximizes the frequency of target functions computation, called the \textit{computing capacity}, per network use at the sinks is known as the \textit{network computation} problem \cite{Appu}. This subsumes the network coding problem as a special case. Environmental monitoring in an industrial unit is an application of network computation where relevant parameter may include temperature and level of exhaust gases which may assist in preventing fire and poisoning due to toxic gases respectively. 

A simple way to perform network computation is to communicate all the messages relevant to the function required at each sink using either network coding or routing. This is not only highly inefficient in terms of bandwidth usage and power consumption but also undesirable in certain settings. For example, in an election, who voted whom is to be kept confidential but the sum total of votes received by each candidate is to be publicized. An efficient way is that function computation be performed \textit{in-network}, i.e., in a distributed manner. The intermediate nodes on the paths between the sources and the sinks perform network coding and communicate coded messages such that the sinks may compute their desired functions without having to know the value of the arguments. 

In \cite{GiriKr}, bounds on rate of computing symmetric functions (invariant to argument permutations) of data collected by sensors in a wireless sensor network at a sink node were presented. The notion of min-cut bound for the network coding problem \cite{Yeung} was extended to the function computation problem in a directed acyclic network with multiple sources and one sink in \cite{Appu}. The case of directed acyclic network with multiple sources, multiple sinks and each sink demanding the sum of source messages was studied in \cite{Dey}; such a network is called a sum-network. Relation between linear solvability of multiple-unicast networks and sum-networks was established. Furthermore, insufficiency of scalar and vector linear network codes to achieve computing capacity for sum-networks was shown. Coding schemes for computation of arbitrary functions in directed acyclic network with multiple sources, multiple sinks and each sink demanding a function of source messages were presented in \cite{Dey2}. In \cite{Appu2}, routing capacity, linear coding capacity and nonlinear coding capacity for function computation in a multiple source single sink directed acyclic network were compared and depending upon the demanded functions and alphabet (field or ring), advantage of linear network coding over routing and nonlinear network coding over linear network coding was shown.

Connections between matroids and network coding problems were first explored in \cite{Dough} wherein matroidal networks were characterized and a construction method to obtain matroidal networks from matroids was given. It was shown in \cite{Dough} that if a network admits a scalar linear solution, then the network is matroidal with respect to represenatable matroid. The converse, i.e., if a network is matroidal with respect to a representable matroid, then it admits a scalar linear solution was given in \cite{KimMed}. 
The construction procedure to obtain network from matroids given in \cite{Dough} reflects all the independencies but not all the dependencies of the matroids. This problem was addressed in\cite[Sec.~VI]{ICRel} wherein a method to construct a network from a matroid was given via an index coding problem; the resultant network reflects all the dependencies and independencies of the matroids and it was shown that a scalar (vector) linear solution exists for the network if and only if the matroid has a linear (multilinear) representation. Also in \cite{ICRel} and \cite{ICRel2}, relation between network and index coding was studied. Similar relation between network computation problems and functional index coding problems (a generalization of index coding problem proposed in \cite{FICP}) was established in \cite{FNC-FIC}.

\subsection{Contributions and Organization}
In this paper, we explore the relationship of network computation problems and matroid representations. The contributions of this paper are as follows:
\begin{enumerate}
\item In Section~\ref{sec_main}, we give a modified definition of matroidal networks to fit the requirements of network computation problems; the modified definition subsumes the original definition of \cite{Dough} as a special case (Remark~\ref{rem_2}).  
\item In Theorem~\ref{thm1}, we show that a scalar linear network code for a given network computation problem with linear functions demanded by sinks exists if and only if the network is matroidal with respect to a representable matroid whose representation satisfies certain constraints imposed by the network computation problem. This generalizes \cite[Th.~12]{KimMed} which states that a network coding problem admits a scalar linear solution if and only if it is matroidal with respect to a representable matroid. 
\item In Section~\ref{sec_fd}, we show connection between \textit{functional dependency relations} (FD-relations) and network computation problems with possibly nonlinear function demands. In Proposition~\ref{prop_1}, we show that a functional representation of an FD-relation (determined by the network computation problem) can be used to obtain nonlinear network codes. This generalizes \cite[Proposition~12]{FDNC} which states that a network coding problem admits a scalar solution if and only if the corresponding FD-relation has a functional representation. 
\end{enumerate}
In Section~\ref{sec_prelims}, relevant preliminaries of network computation problem and matroid theory are given. We conclude the paper with a summary of work presented in Section~\ref{sec_disc}.

\section{Network Model} \label{sec_prelims}
A brief overview of network computation problems and matroids are presented in this section. A $q$-ary finite field is denoted by $\mathbb{F}_q$ and the set $\{1,2,\ldots,n\}$ is denoted by $[n]$, for some positive integer $n$. The power set of a set $S$ is denoted by $2^S$. The column vector of length $N$ whose $n$th component is one and all other components are zeros is denoted as $\epsilon_{n,N}$. An $N\times N$ identity matrix is denoted by $\mathrm{I}_{N\times N}$ and an $m\times n$ all zero matrix is denoted by $\mathbf{0}_{m\times n}$.

\subsection{Network Computation}
A network is represented by a finite directed acyclic graph $\mathcal{N}=(V,\mathcal{E})$, where $V$ is the set of nodes and $\mathcal{E}=\tilde{E}\cup E\cup \hat{E}$ is the set of directed error-free links (edges), where the edges in $\tilde{E}$ correspond to the source messages generated in the network, the edges in $E$ correspond to the links between the nodes in the network, and the edges in $\hat{E}$ correspond to the demands of the sink nodes. For a node $w\in V$, $In(w)$ is the set of messages generated (if any) by node $w$ and the incoming links to it from other nodes, and $Out(w)$ is the set of outgoing links to other nodes and the function demanded (if any) by $w$. For an edge $e=(u,v)\in E$ from a node $u$ to $v$, $u$ and $v$ are called, respectively, its tail and head and $In(e)=In(u)$, i.e., $In(e)$ is the set of edges which terminate at the node at which $e$ originates. The network may have multiple source nodes and each may generate multiple messages. The source messages are represented by tailless edges $\tilde{e}_k\in \tilde{E}$ that terminate at a source node. The total number of messages generated in the network is $K=|\tilde{E}|$ and are denoted by random variables $X_1,X_2,\ldots,X_K$, where, for every $k\in [K]$, $X_k$ is uniformly distributed over $\mathbb{F}_q$. Let $X=(X_1,\ldots,X_K)$ be the row vector of source messages and $\mathcal{X}=\{X_1,\ldots,X_K\}$ be the set of source messages. 
Each link $e\in E$ can carry one $\mathbb{F}_q$ element per use, i.e., capacity of each link is $1$, and $Y_e$ is the associated random variable. Note that for a source edge $\tilde{e}_k\in \tilde{E}$, the associated random variable is $Y_{\tilde{e}_k}=X_k$. The set of sink nodes is denoted by ${T}$. Without loss of generality, we assume that each sink demands only one function of source messages. If a sink demands $N\,(>1)$ functions, then such a sink may be replaced by $N$ sinks, each demanding one function but receiving the same incoming information. A sink node $t$ requests a function $g_t(X)$, where $g_{t}:\mathbb{F}_q^{K}\rightarrow \mathbb{F}_q$. 
The demands $g_t(X)$ of the sink node $t$ is denoted by a headless edge $\hat{e}_t\in \hat{E}$ originating at $t$. Let $G_T=\{g_t:t\in T\}$.

A network computation problem $\mathcal{F}(\mathcal{N}(V,\mathcal{E}),\mathcal{X},G_T)$ is specified by the underlying network, the message set, and the set of sink demands.

A network code $\{F_e:e\in E\}\cup \{D_t:t\in {T}\}$ for a network computation problem $\mathcal{F}$ is an assignment of a \textit{global encoding kernel} $F_e:\mathbb{F}_q^K\rightarrow \mathbb{F}_q$ to each edge $e\in E$ and a decoding function $D_t:\mathbb{F}_q^{|In(t)|}\rightarrow \mathbb{F}_q$ to each sink $t\in {T}$. For any edge $e\in E$, $F_e$ maps $X$ to $Y_e$ (and thus, the distribution of $Y_e$ depends upon the network code), i.e., $Y_e=F_e(X)$ is the data that flows on edge $e$. For any sink $t\in {T}$, the decoding map $D_t$ takes as input the data on the incoming edges, $(Y_{e'})_{e'\in In(t)})$, and outputs $g_t(X)$, i.e.,
\begin{align}\label{eq_dec_y}
D_t\left( (Y_{e'})_{e'\in In(t)} \right) = D_t\left( (F_{e'}(X))_{e'\in In(t)} \right) = g_t(X).
\end{align}

For every tailless edge $\tilde{e}_k\in\tilde{E}$ denoting the source message $X_k$, $F_{\tilde{e}_k}(X)=X_k$ is taken to be the global encoding kernel. For every headless edge $\hat{e}_t\in \hat{E}$ denoting the demand $g_t(X)$ of sink $t\in T$, $F_{\hat{e}_t}(X)=g_t(X)$ is taken to be the global encoding kernel. 

\begin{remark} \label{rem_1} \begin{enumerate}[(a)]
\item A code is said to be linear if all the global encoding kernels of edges in $E$ are linear, i.e., data on the outgoing edges of each node is a linear combination of data on the incoming edges. Also, the global encoding kernels of the outgoing edges of each node is a linear combination of the global encoding kernels of the incoming edges of that node. Moreover, the global encoding kernel for an edge $e\in \mathcal{E}$ can be represented by a length $K$ column vector $F_e$ over $\mathbb{F}_q$ such that $Y_e=F_e(X)=X\cdot F_e$ for all $e\in E$, where $F_e$ is called the global encoding vector of $e$.
\item Similarly, the global encoding kernel of edge $\tilde{e}_k$, $k\in [K]$, can be represented by $\epsilon_{k,K}$ so that $F_{\tilde{e}_k}(X)=X\cdot F_{\tilde{e}_k}=X\cdot\epsilon_{k,K}=X_k$. Note that the matrix obtained by juxtaposing global encoding vectors of $\tilde{e}_1,\ldots,\tilde{e}_K$ is the identity matrix $\mathrm{I}_{K\times K}$.
\item If the sink demands are linear then they can be represented by $K$ length columns vectors $g_t$ such that $g_t(X)=X\cdot g_t$.
\item If all the sink demands are linear and the network computation problem admits a linear solution, then all the decoding maps will also be linear. For a sink $t$, the decoding map can be represented using a length $|In(t)|$ column vector $D_t$ such that $g_t(X)=X\cdot g_t=D_t\left( (Y_{e'})_{e'\in In(t)} \right)= ( Y_{e'} )_{e'\in In(t)}\cdot D_t=X\cdot (F_{e'})_{e'\in In(t)}\cdot D_t$.
\end{enumerate}
\end{remark}

Another way to specify a network code is to list the \textit{local encoding kernel} $f_e:\mathbb{F}_q^{|In(e)|}\rightarrow \mathbb{F}_q$ of each edge $e\in E$ and the decoding functions of the sinks. That is, a network code $\{f_e:e\in E\}\cup \{D_t:t\in {T}\}$ for a network computation problem $\mathcal{F}$ is an assignment of a local encoding kernel $f_e$ to each edge $e\in E$ and a decoding function $D_t$ to each sink $t\in {T}$. For any $e\in E$, $f_e$ takes in $(Y_{e'})_{e'\in In(e)}$ as input argument and outputs $Y_e$, i.e.,
\begin{align*}
f_e:(Y_{e'})_{e'\in In(e)}\longmapsto Y_e.
\end{align*}
Given the local encoding kernels, the global encoding kernels for each edge can be defined by induction on an ancestral ordering of the edges in the graph as follows (such an ordering always exists for acyclic graphs). For every tailless edge $\tilde{e}_k\in\tilde{E}$ denoting the source message $X_k$, take $F_{\tilde{e}_k}(X)=X_k$ to be the global encoding kernel. Then, for any edge $e\in E$, 
\begin{align} \label{eq_f_2_F}
F_e(X)=f_e\left((F_{e'}(X))_{e'\in In(e)}\right).
\end{align}

\subsection{Matroids}
Now we review relevant concepts of matroid theory. Comprehensive treatment of matroid theory can be found in \cite{Oxley,Welsh}. 

A matroid $\mathcal{M}({S},r)$ comprises a \textit{ground set} ${S}$ and a \textit{rank function} $r:2^{S}\rightarrow \mathbb{Z}$ that satisfies the following axioms:
\begin{enumerate}[(R1)]
\item $0\leqslant r(A) \leqslant |A|$ for all $A\subseteq S$;
\item $r(A)\leqslant r(B)$ for all $A\subseteq B\subseteq S$;
\item $r(A\cup B)+r(A\cap B)\leqslant r(A)+r(B)$ for all $A,B\subseteq S$.
\end{enumerate}%

The rank $r_{\mathcal{M}}$ of the matroid is $r_{\mathcal{M}}=r(S)$. A subset $A$ of the ground set is said to be \textit{independent} if $r(A)=|A|$, if not then it is said to be \textit{dependent}. A subset $B$ of $S$ is said to be a \textit{basis} if $r(B)=|B|=r_{\mathcal{M}}$. A basis is a maximal independent set, i.e., adding one more element will make it dependent.

Let $M$ be an $m\times n$ matrix over $\mathbb{F}_q$ whose columns are denoted by $M_i$, $i=1,2,\ldots,n$. Let $S=\{1,2,\ldots,n\}$ and let $r:2^S\rightarrow\mathbb{Z}$ be a function such that, for every $A\subseteq S$, $r(A)$ is the rank of the $m\times |A|$ submatrix of $M$ with columns indexed by $A$, i.e., if $A=\{i_1,i_2,\ldots,i_{|A|}\}\subseteq S$, then $r(A)=\mathrm{rank}[M_{i_1},M_{i_2},\cdots,M_{i_{|A|}}]$. In other words, a subset $A$ of $S$ is independent if and only if the columns indexed by it are linearly independent. Hence, $r_{\mathcal{M}}=\mathrm{rank}(M)$. The matroid defined above is called the \textit{vector matroid of $M$} and is denoted by $\mathcal{M}(M)$; the ground set of this matroid is the set of column indices $[n]$ of the matrix $M$. 

A matroid $\mathcal{M}(S,r)$ is said to be \textit{representable over $\mathbb{F}_q$} if there exist column vectors $v_1,v_2,\ldots,v_{|S|}$ over $\mathbb{F}_q$ and a bijective map $\phi:S\rightarrow V$ that preserves rank, i.e., for all $A=\{s_1,s_2,\ldots,s_{|A|}\}\subseteq S$, $\mathrm{rank}[v_{s_1},v_{s_2},\cdots,v_{s_{|A|}}]=r(A)$. The set of vectors $v_1,v_2,\ldots,v_{|S|}$ is said to form a representation of $\mathcal{M}(S,r)$ and can be described by a matrix obtained by juxtaposing $v_1,v_2,\ldots,v_{|S|}$. Thus, matroid representation can be seen as assignment of vectors to the matroid ground set elements such that a subset of vectors is independent if and only if the subset of ground set elements they correspond is an independent set in the matroid. A matroid can have several representations over a field. Note that the length $d$ of the column vectors forming a representation of a matroid $\mathcal{M}$ over $\mathbb{F}_q$ must be at least $r_{\mathcal{M}}$, so that at least $r_{\mathcal{M}}$ independent vectors exist in $\mathbb{F}_q^{d}$.

\begin{example}
Consider a matroid on the ground set $S=\{1,2,3\}$ and rank function $r=\min\{|A|,2\}$ for all $A\subseteq S$, i.e., any subset of cardinality at most two is independent. Representations $M_2$ and $M_3$ of $\mathcal{M}(S,r)$ over $\mathbb{F}_2$ and $\mathbb{F}_3$ respectively are given below.

\vspace{-10pt}
\begin{footnotesize}
\begin{align*}
M_2= \bordermatrix{&1&2&3\cr
                & 0 & 1 & 1\cr
                & 1 & 1 & 0}
\quad                
M_3= \bordermatrix{&1&2&3\cr
                & 1 & 0 & 2\cr
                & 0 & 1 & 2\cr
                & 1 & 1 & 1\cr
                & 2 & 2 & 2}.
\end{align*}
\end{footnotesize}
\vspace{-8pt}

A matroid on $n$ elements ground set $S$ and rank function $r(A)=\min\{|A|,k\}$ for all $A\subseteq S$ and some $k\leqslant n$ is called a \textit{uniform matroid} and is denoted by $U_{k,n}$. The matroid given above is $U_{2,3}$.
\hfill $\square$
\end{example}

\section{Network Computation Problems and Matroids} \label{sec_main}
We now define matroidal networks in the context of network computation problem and prove the main result that relates scalar linear solutions for network computation problems and matroid representation.

\begin{definition} \label{def_matnet}
Let $\mathcal{F}(\mathcal{N}(V,\mathcal{E}),\mathcal{X},G_T)$ be a network computation problem with $K$ source messages and $\mathcal{E}=\tilde{E}\cup E\cup \hat{E}$. Let $\mathcal{M}(S,r)$ be a matroid. Then, the network $\mathcal{N}$ is \textit{matroidal} with respect to the matroid $\mathcal{M}$ if there exists a map $f:\mathcal{X}\cup E\cup G_T\rightarrow S$, called the \textit{network-matroid map}, from the set of edges to the ground set of the matroid that satisfies the following conditions:
\begin{enumerate}[(M1)]
\item $f$ is one-to-one on $\mathcal{X}$;
\item $f(\mathcal{X})$ is independent;
\item $r \left( f(\,In(v)\,) \right) = r \left( f(\,In(v)\cup Out(v)\,) \right)$ $\forall v\in V$.
\end{enumerate}
\end{definition}

Condition (M1) ensures that messages are assigned different matroid ground set elements and condition (M2) ensures that these messages correspond to an independent set. Condition (M3) ensures that the outgoing edges of every node in the network are dependent on the incoming edges of the node.

\begin{remark} \label{rem_2}
Note that in a network coding problem, $\mathcal{X}\cup G_T=\mathcal{X}$ since each sink demands some message from the set $\mathcal{X}$. The network-matroid map in this case simplifies to $f:\mathcal{X}\cup E\rightarrow S$. This is same as \cite[Definition~V.1]{Dough}, and thus is a special case of Definition~\ref{def_matnet} given above. 
\end{remark}

We now present the main result connecting scalar linear solutions of a network computation problem to representable matroids. 

\begin{theorem} \label{thm1}
Let $\mathcal{F}(\mathcal{N}(V,\mathcal{E}),\mathcal{X},G_T)$ be a network computation problem with $K$ source messages, $\mathcal{E}=\tilde{E}\cup E\cup \hat{E}$, and each sink demands a linear combination of source messages. Then, $\mathcal{F}$ admits a scalar linear solution over $\mathbb{F}_q$ if and only if $\mathcal{N}$ is matroidal with respect to a matroid $\mathcal{M}$ which is representable over $\mathbb{F}_q$ and at least one of its representations $M\in \mathbb{F}_q^{m\times n}$, $m\geqslant K$ and $n\geqslant m$, satisfies the following constraints:
\begin{enumerate}[(C1)]
\item $M$ contains an $m\times K$ submatrix of the form $\begin{bsmallmatrix} \mathrm{I}_{K\times K}\\ \mathbf{0}_{m-K\times K} \end{bsmallmatrix}$;
\item $M$ contains an $m\times |T|$ submatrix of the form $\begin{bsmallmatrix}g_{t_1}\,g_{t_2}\,\ldots\,g_{t_{|T|}}\\ \mathbf{0}_{m-K\times |T|}\end{bsmallmatrix}$.
\end{enumerate}
\end{theorem}

\begin{figure*}
\tiny
\begin{align}
M= \bordermatrix{&F_{\tilde{e}_1}&F_{\tilde{e}_2}&F_{\tilde{e}_3}& F_{\tilde{e}_4}& F_{e_1} &F_{e_2} &F_{e_3}&F_{e_4}&F_{e_5}&F_{e_6}&F_{e_7}&F_{e_8}&F_{e_9}&F_{e_{10}}&F_{e_{11}}&F_{e_{12}}&F_{e_{13}}&F_{e_{14}}&F_{e_{15}}&F_{e_{16}}&F_{\hat{e}_1}&F_{\hat{e}_2}&F_{\hat{e}_3}& F_{\hat{e}_4}\cr
& 1 & 0 & 0 & 0 & 1 & 1 & 0 & 0 & 1 & 0 & 0 & 1 & 0 & 0 & 1 & 1 & 0 & 1 & 1 & 0 & 1 & 1 & 0 & 0\cr
& 0 & 1 & 0 & 0 & 0 & 1 & 0 & 0 & 0 & 0 & 0 & 0 & 0 & 0 & 0 & 1 & 0 & 0 & 1 & 0 & 0 & 0 & 1 & 1\cr
& 0 & 0 & 1 & 0 & 0 & 0 & 1 & 1 & 0 & 0 & 1 & 0 & 1 & 1 & 0 & 0 & 1 & 0 & 1 & 1 & 1 & 0 & 1 & 0\cr
& 0 & 0 & 0 & 1 & 0 & 0 & 1 & 0 & 0 & 0 & 0 & 0 & 1 & 0 & 0 & 0 & 0 & 0 & 1 & 0 & 0 & 1 & 0 & 1}
\label{eq_nc2mat}
\end{align}
\vspace{-10pt}
\end{figure*}

\begin{figure*}
\tiny
\begin{align}
M'= \bordermatrix{&1&2&3& 4& 5 &6 &7&8&9&10&11&12&13&14&15&16&17&18&19&20&21&22&23&24\cr
& 1 & 0 & 0 & 0 & 1 & 1 & 0 & 0 & 1 & 1 & 0 & 1 & 1 & 0 & 1 & 2 & 0 & 1 & 2 & 0 & 1 & 1 & 0 & 0\cr
& 0 & 1 & 0 & 0 & 1 & 2 & 0 & 0 & 1 & 2 & 0 & 1 & 2 & 0 & 1 & 1 & 0 & 1 & 1 & 0 & 0 & 0 & 1 & 1\cr
& 0 & 0 & 1 & 0 & 0 & 0 & 1 & 1 & 0 & 1 & 1 & 0 & 2 & 1 & 0 & 1 & 1 & 0 & 2 & 1 & 1 & 0 & 1 & 0\cr
& 0 & 0 & 0 & 1 & 0 & 0 & 2 & 1 & 0 & 2 & 1 & 0 & 1 & 1 & 0 & 2 & 1 & 0 & 1 & 1 & 0 & 1 & 0 & 1}
\label{eq_mat2nc}
\end{align}
\vspace{-10pt}
\end{figure*}

\begin{proof} Let $\tilde{e}_1,\ldots,\tilde{e}_K,e_1,\ldots,e_{|E|},\hat{e_1},\ldots,\hat{e}_{|T|}$ be an ancestral ordering of edges in the network. We first show that a scalar linear solution for $\mathcal{F}$ describes a representable matroid with respect to which $\mathcal{N}$ is matroidal.

\noindent \textit{Scalar linear solution implies $\mathcal{N}$ is matroidal with respect to a representable matroid:} Let $\{F_e:e\in \mathcal{E}\}$ and $\{D_t:t\in {T}\}$, respectively, be the set of global encoding vectors and decoding functions of a scalar linear network code for $\mathcal{F}$ over a field $\mathbb{F}_q$. Let $M$ be a $K\times |\mathcal{E}|$ matrix over $\mathbb{F}_q$ formed by juxtaposing the global encoding vectors of all the edges in $\mathcal{E}$, i.e.,
\begin{align*}
M=& [ F_{\tilde{e}_1}\,\cdots\,F_{\tilde{e}_K}  \, &| \,F_{e_1}\,\cdots\,F_{e_{|E|}}  \, &| \, F_{\hat{e}_1}\,\cdots\,F_{\hat{e}_{|T|}} ] \qquad
\\
=& [ \mathrm{I}_{K\times K} \, &| \, F_{e_1}\,\cdots\,F_{e_{|E|}} \, &| \, g_{t_1}\,\cdots\,g_{t_{|T|}} ].\qquad
\end{align*}
Let $S=\{1,2,\ldots,|\mathcal{E}|\}$ and $\mathcal{M}$ be the vector matroid of $M$ and $r$ be its rank function. Let $f:\mathcal{X}\cup E \cup G_T\rightarrow S$ be defined as follows: 
\begin{align*}
f(X_k)&=k,     & k=1,2,\ldots,K;\\
f(e_i)&= K+i,  & i=1,2,\ldots,|E|;\\
f(g_{t_j})&=K+|E|+j,  & j=1,2,\ldots,|T|.
\end{align*}
We will verify that $f$ satisfies (M1)-(M3) so that $\mathcal{N}$ is matroidal with respect to the vector matroid of $M$.

The function $f$ is one-to-one on $\mathcal{X}$ (distinct elements are assigned to each message) thus satisfying (M1). It also satisfies (M2) since $f(\{X_1,X_2,\ldots,X_K\})=\{1,2,\ldots,K\}$ and $r(\{1,2,\ldots,K\})=K$ since first $K$ columns of $M$ are linearly independent. Let $v$ be an arbitrary node in the network and let $e\in \mathcal{E}$ be an outgoing edge of $v$, i.e., $e\in Out(v)$. Then, $F_e$ is a linear combination of the global encoding vectors $F_{e'}$, $e'\in In(v)$ (by Remark~\ref{rem_1}), and consequently $r(f(In(v)\cup Out(v)))=\mathrm{rank}\left[(F_{e'})_{e'\in In(v)},(F_{e})_{e\in Out(v)}\right]=\mathrm{rank}\left[(F_{e'})_{e'\in In(v)}\right]=r(f(In(v)))$. Here we have used the fact that $In(v)$ is the set containing the messages generated by $v$ and incoming links from other nodes and hence $f(In(v))$ is a subset of $S$ containing ground set elements corresponding to the messages generated by $v$ and those corresponding to incoming links to $v$. The vectors $F_{e'},\;e'\in In(v)$, are assigned to these elements. Similarly, vectors $F_{e},\;e\in Out(v)$ are assigned to the ground set elements in $f(Out(v))$, where $Out(v)$ includes the demand and outgoing links of $v$. Thus, $f$ satisfies (M3). Clearly $M$ satisfies (C1) and (C2) with $m=K$ and $n=|\mathcal{E}|$.

We now prove the converse.

\noindent \textit{$\mathcal{N}$ is matroidal with respect to a representable matroid implies $\mathcal{F}$ admits a scalar linear solution:} Let $\mathcal{N}$ be matroidal with respect to a representable matroid $\mathcal{M}$ and $M\in \mathbb{F}_q^{m\times n}$, $m\geqslant K$ and $n\geqslant m$ be a representation that satisfies (C1) and (C2). Since $M$ is a representation of $\mathcal{M}$, $r(S)=\mathrm{rank}(M)$ (since a set of column indices is independent if and only if the corresponding columns are linearly independent). Let $r(S)=\mathrm{rank}(M)=m$; if not then redundant rows can be dropped without changing the dependencies or independencies of the matroid. Since the network is matroidal, there exists a network-matroid map $f:\mathcal{X}\cup E \cup G_T\rightarrow S$, where $S=[n]$, that satisfies (M1)-(M3). Without loss of generality, let columns $1$ through $K$ of $M$ be of the form $\begin{bsmallmatrix} \mathrm{I}_{K\times K}\\ \mathbf{0}_{m-K\times K} \end{bsmallmatrix}$ and $f(X_k)=k$ for $k=1,2,\ldots,K$. Let columns $f(g_{t_1}),f(g_{t_2}),\ldots,f\left(g_{t_{|T|}}\right)$ of $M$ form the submatrix $\begin{bsmallmatrix}g_{t_1}\,g_{t_2}\,\ldots\,g_{t_{|T|}}\\ \mathbf{0}_{m-K\times |T|}\end{bsmallmatrix}$. 

Add a dummy source node that is not connected to any node in the network and generates dummy source messages $X_{K+1},\ldots,X_m$. The modified network has $m$ source messages and hence the length of global encoding vectors will also be $m$. Taking into account the dummy source messages, the global encoding vector of an edge $\hat{e}_{t_j}$, $t_j\in T$, will now be $F_{\hat{e}_{t_j}}=\begin{bsmallmatrix}g_{t_j}\\ \mathbf{0}_{m-K\times 1}\end{bsmallmatrix}=M_{f(g_{t_j})}$. Assign the global encoding vectors to the edges $e\in \mathcal{E}$ as follows:
\begin{align*}
F_{\tilde{e}_k}&=M_{f(X_k)}=M_k,     & k=1,2,\ldots,K;\\
F_{e_i}&= M_{f(e_i)},  & i=1,2,\ldots,|E|.
\end{align*}
We will verify that the above choice of global encoding vectors satisfies all the sink demands by obtaining a decoding map $D_t$ for each sink $t\in T$. 

Since $f$ satisfies (M3), $r(f(In(v)))=r(f(In(v)\cup Out(v)))$ for every $v\in V$. That is, for each node $v\in V$ and each edge $e\in Out(v)\cap E$ (i.e., edges from $v$ to other nodes in the networks and not representing demands), $M_{f(e)}$ is a linear combination of $\{M_{f(e')}$, $e'\in In(v)\}$. Consequently, by the above choice of encoding vectors, $F_e$ is a linear combination of $\{F_{e'}$, $e'\in In(v)\}$. 

For a sink node $t\in T$, $M_{f(g_t)}$, the global encoding vector of $\hat{e}_t$, is a linear combination of $M_{f(e')}$, $e'\in In(v)$, because the edge $\hat{e}_t$ is in $Out(t)$ and $r(f(In(t)))=r(f(In(t)\cup Out(t)))$ by (M3). Thus, $F_{\hat{e}_{t_j}}=M_{f(g_t)}$ is a linear combination of $F_{e'}$, $e'\in In(t)$. That is, there exists a columns vector $D_t$ of length $|In(t)|$ such that $F_{\hat{e}_{t_j}}=M_{f(g_t)}=\begin{bsmallmatrix}g_{t_j}\\ \mathbf{0}_{m-K\times 1}\end{bsmallmatrix}=(F_{e'})_{e'\in In(t)}\cdot D_t$ (Remark~\ref{rem_1}(d)). The vector $D_t$ is the decoding vector for sink node $t$. Removing the dummy messages and deleting the last $m-K$ rows of global encoding maps defined above, we get a network code, i.e., a global encoding vector of length $K$ for each edge in $E$ and a decoding vector for each sink in $T$, for the network coding problem $\mathcal{F}$ over the network $\mathcal{N}$ (without the dummy source node and with only $K$ messages). 

Thus, a representation, satisfying (C1) and (C2), of a matroid with respect to which $\mathcal{N}$ is matroidal gives a scalar linear solution for the network computation problem $\mathcal{F}$. 
\end{proof}

%

\begin{example}
Consider the network computation problem given in Fig.~\ref{fig_fncmat}. 
\begin{figure}[h]
\vspace{-10pt}
\centering
\includegraphics[scale=0.52]{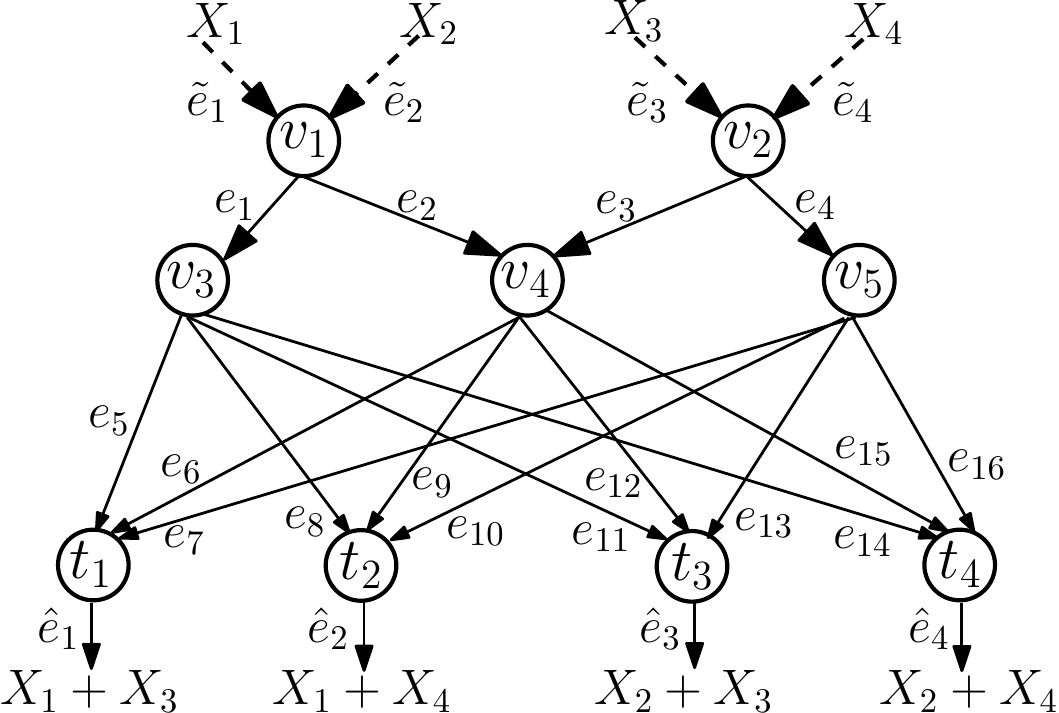}
\caption{A network coding problem.}
\label{fig_fncmat}
\vspace{-8pt}
\end{figure}

The global encoding vectors of a linear network code over $\mathbb{F}_2$ are given in \eqref{eq_nc2mat}. Decoding takes place as follows:

\vspace{-12pt}
\begin{footnotesize}
\begin{align*}
g_{t_1}(X)&=Y_{e_5}+Y_{e_7}=X(F_{e_5}+F_{e_7})=X_1+X_3;\\
g_{t_2}(X)&=Y_{e_8}+Y_{e_9}+Y_{e_{10}}=X(F_{e_8}+F_{e_9}+F_{e_{10}})=X_1+X_4;\\
g_{t_3}(X)&=Y_{e_{11}}+Y_{e_{12}}+Y_{e_{13}}=X(F_{e_{11}}+F_{e_{12}}+F_{e_{13}})=X_2+X_3;\\
g_{t_4}(X)&=Y_{e_{14}}+Y_{e_{15}}+Y_{e_{16}}=X(F_{e_{14}}+F_{e_{15}}+F_{e_{16}})=X_2+X_4.
\end{align*}
\end{footnotesize}
\vspace{-12pt}

Let $\mathcal{M}(M)$ be the vector matroid of matrix $M$ given in \eqref{eq_nc2mat} with ground set $S=[24]$. Let the network-matroid map $f$ be defined as follows:
\begin{align}
\nonumber f(X_k)&=k, \quad &k=1,2,3,4;\\
f(e_i)&=4+i, \quad &i=1,2,\ldots,16; \label{eq_map} \\
\nonumber f(g_{t_j})&=20+j, \quad &j=1,2,3,4.
\end{align}
Then, it can be verified that $f$ satisfies (M1)-(M3) and the network is matroidal with respect to the representable matroid $\mathcal{M}(M)$. \hfill $\square$
\end{example}

\begin{example}
Let $\mathcal{M}(M')$ be the vector matroid of matrix $M'$ over $\mathbb{F}_3$ given in \eqref{eq_mat2nc} with ground set $S=[24]$. The network in Fig.~\ref{fig_fncmat} is matroidal with respect to $\mathcal{M}(M')$ with the network-matroid map $f$ defined as in \eqref{eq_map} and the matrix $M'$ satisfies (C1) and (C2). Then, by Theorem~\ref{thm1}, the following choice of global encoding kernels gives a network code over $\mathbb{F}_3$ for the network computation problem of Fig.\ref{fig_fncmat}: 
\begin{align*}
F_{\tilde{e}_k}&=M'_{f(X_k)}=M'_k,     & k=1,2,3,4;\\
F_{e_i}&= M'_{f(e_i)},  & i=1,2,\ldots,16.
\end{align*}
The decoding is performed as follows:
\begin{footnotesize}
\begin{align*}
g_{t_1}(X)&=Y_{e_5}+Y_{e_6}+Y_{e_7}=X(F_{e_5}+F_{e_6}+F_{e_7})=X_1+X_3;\\
g_{t_2}(X)&=Y_{e_8}+Y_{e_9}+Y_{e_{10}}=X(F_{e_8}+F_{e_9}+F_{e_{10}})=X_1+X_4;\\
g_{t_3}(X)&=Y_{e_{11}}+Y_{e_{12}}+Y_{e_{13}}=X(F_{e_{11}}+F_{e_{12}}+F_{e_{13}})=X_2+X_3;\\
g_{t_4}(X)&=Y_{e_{14}}+Y_{e_{15}}+Y_{e_{16}}=X(F_{e_{14}}+F_{e_{15}}+F_{e_{16}})=X_2+X_4.
\end{align*}
\end{footnotesize}
Thus, a matroid representation satisfying (C1) and (C2) gives a scalar linear network code. \hfill $\square$
\end{example}

\section{Network Computation and FD-Relations} \label{sec_fd}
In the preceding section, we established connections between network computation problems and matroid representation. There are two issues with this connection. Firstly, the network does not dictate that every set of edges carry dependent or independent information; only the source message, and hence the corresponding tailless edges, need to be independent and the data on outgoing edges of each node be completely determined by the data on its incoming edges. Thus, only the ground set elements corresponding to source messages need to be independent, which is reflected by (M2), and the set of elements corresponding to incoming and outgoing edges of every node be a dependent set, which is reflected by (M3). But the matroid will have additional dependencies and independencies since every subset of the ground set should be either dependent or independent. Secondly, this connection is limited to linear network codes for network computation problems with linear sink demands. 

The issues specified above can be circumvented by using a FD-relation to characterize a network. FD-relations were defined by \cite{FD} and matroids were shown to be a special case of FD-relation \cite[eq.~(4)]{FD}. FD-relations find application in relational database theory. Connections between FD-relations and network coding problems were explored in \cite[Sec.~VI]{FDNC} to circumvent the problem of additional dependencies and independencies that arise in the matroid-network coding problems connection.

\begin{definition}
Let $N$ be a finite set and let $\mathcal{Q}(N)$ denote the set of all ordered pair of subsets of $N$, i.e., $\mathcal{Q}(N)=\{(I,J):I,J\subseteq N\}=2^N\times 2^N$. Then, $Q\subseteq \mathcal{Q}(N)$ is said to be a FD-relation on $N$ if and only if it satisfies the following conditions:
\begin{enumerate}[(FD1)]
\item If $I\subseteq J \subseteq N$, then $(I,J)\in Q$;
\item If $(I,J)\in Q$ and $(J,K)\in Q$, then $(I,K)\in Q$;
\item If $(I,J)\in Q$ and $(J,K)\in Q$, then $(I,J\cup K)\in Q$.
\end{enumerate}
\end{definition}
\noindent For any $(I,J)\in Q$, $J$ is said to \textit{depend functionally} on $I$. 

Representation of FD-relations was also studied in \cite{FD}. Let, for each $i\in N$, $\Phi_i$ be a map from a nonempty set $B$ to a nonempty set $C$, i.e., $\Phi_i:B\rightarrow C$. For a subset $I$ of $N$, define $\Phi_I:B\rightarrow C^{|I|}$ by $f_I(b)=(\Phi_i(b))_{i\in I}\in C^{|I|}$ (with $C^{|I|}=\{\emptyset\}$ for $I=\emptyset$). Then, the set $\{\Phi_i:i\in N\}$ forms a \textit{functional representation} \cite[Sec.~VI]{FDNC},\cite[Example~3 and Remark~3]{FD} of $Q$ if and only if for every $(I,J)\in Q$, there exists a function $\Psi_I^J:C^{|I|}\rightarrow C^{|J|}$ such that
\begin{align} \label{eq_fdrep}
\Phi_J=\Psi_I^J\circ \Phi_I.
\end{align}
In other words, a functional representation of $Q$ is an assignment of functions from a set $B$ to a set $C$ to the elements of $N$ that satisfy \eqref{eq_fdrep}.

Let $\mathcal{F}(\mathcal{N}(V,\mathcal{E}),\mathcal{X},G_T)$ network computation problem with possibly nonlinear sink demands and $K$ messages. For a node $v$, let $Out'(v)=Out(v)\backslash \hat{E}$, i.e, unlike $Out(v)$, $Out'(v)$ does not include the headless edges denoting the demand (if any) of node $v$. Let an FD-relation $Q_\mathcal{E}$ on the set of edges $\mathcal{E}$  be
\begin{align} \label{eq_fdrel}
Q_\mathcal{E}\!=\!\{(In(v),Out'(v))\!:\!v\in V\}\cup\{(In(t),\hat{e}_t)\!:\!t\in T\}.
\end{align}
This FD-relation reflects exactly the dependencies of the network and contains no additional dependencies. 

\begin{remark} Any functional representation of $Q_\mathcal{E}$ is an assignment of functions to the edges.
\begin{enumerate}[(a)]
\item For an edge $e\in \mathcal{E}$, the function $\Phi_e$ is the global encoding kernel of $e$.
\item For a node $v\in V$, let $I=In(v)$ and $J=Out'(v)$. Then, the function $\Psi_I^J$ is the set of local encoding kernels of the outgoing edges of $v$ (from \eqref{eq_fdrep} and \eqref{eq_f_2_F}).
\item For a sink $t\in T$, let $I=In(t)$ and $J=\hat{e}_t$. Then, the function $\Psi_I^J$ is the decoding function of sink $t$ (from \eqref{eq_fdrep} and \eqref{eq_dec_y}).
\end{enumerate}
\end{remark}

We have the following result.

\begin{proposition} \label{prop_1}
The network computation problem $\mathcal{F}$ admits a scalar solution over $\mathbb{F}_q$ if and only if there exists a functional representation $\{\Phi_e:e\in \mathcal{E}\}$ of $Q_\mathcal{E}$ with $B=\mathbb{F}_q^K$ and $C=\mathbb{F}_q$ (i.e., $\Phi_e:\mathbb{F}_q^K\rightarrow \mathbb{F}_q$ for all $e\in \mathcal{E}$) that satisfies the following constraints:
\begin{enumerate}[(C1){$^\prime$}]
\item $\Phi_{\tilde{e}_k}(X)=X_k$ for all $k\in [K]$;
\item $\Phi_{\hat{e}_t}(X)=g_t(X)$ for all $t\in T$.
\end{enumerate}
\end{proposition}

\begin{proof}
Let $\{F_e:e\in {E}\}$, $\{f_e:e\in E\}$, and $\{D_t:t\in {T}\}$, respectively, be the set of global encoding vectors, local encoding vectors, and decoding functions of a scalar network code for $\mathcal{F}$ over $\mathbb{F}_q$. Let $F_{\tilde{e}_k}=X_k$ for all $k\in [K]$ and $F_{\hat{e}_{t}}=g_{t}(X)$ for all $t\in T$. Using \eqref{eq_f_2_F}, we have for a vertex $v\in V$
\begin{align*}
(F_e(X))_{e\in Out'(v)} = \left( f_e((F_{e'}(X))_{e\in In(v)})  \right)_{e\in Out'(v)},
\end{align*}
i.e., $(f_e)_{e\in Out'(v)}:\mathbb{F}_q^{|In(v)|}\rightarrow \mathbb{F}_q^{|Out'(v)|}$. And, for a sink $t$
\begin{align*}
D_t\left( (F_{e'}(X))_{e\in In(t)} \right) = g_t(X).
\end{align*}

Let $\Phi_e=F_e$ for all $e\in \mathcal{E}$. For a node $v\in V$, let $I=In(v)$, $J=Out'(v)$, and $\Psi_I^J=(f_e)_{e\in J}$. For a sink $t\in T$, let $I=In(t)$, $J=\hat{e}_t$, and $\Psi_I^J=D_t$. Then, $\{\Phi_e:e\in \mathcal{E}\}$ forms a functional representation of $Q_\mathcal{E}$ with $B=\mathbb{F}_q^K$, $C=\mathbb{F}_q$ and the function $\Phi_e$s and $\Psi_I^J$s satisfying \eqref{eq_fdrep}.

To prove the converse, assume that a functional representation of $Q_\mathcal{E}$ that satisfies (C1){$^\prime$} and (C2){$^\prime$} is given. That is, $\{\Phi_e:e\in \mathcal{E}\}$, $\{\Psi_I^J:I=In(v),J=Out'(v),v\in V\}$, and $\{\Psi_I^J:I=In(t),J=\hat{e}_t,t\in T\}$ are given and satisfy (C1){$^\prime$} and (C2){$^\prime$}. For each $t\in T$, the network computation problem specifies $F_{\hat{e}_t}=g_t$ which is equal to $\Phi_{\hat{e}_t}$ (by (C2){$^\prime$}). To each edge $e\in \tilde{E}\cup E$, assign $\Phi_e$ as the global encoding kernel, i.e., $F_e=\Phi_e$. For each $v\in V$, $I=In(v)$, and $J=Out'(v)$, let $(f_e)_{e\in J}=\Psi_I^J$. For each $t\in T$, $I=In(t)$, and $J=\hat{e}_t$, let $D_t=\Psi_I^J$. We will verify that this choice of encoding kernels and decoding maps satisfies all the sink demands. 

For each node $v\in V$, $J=Out'(v)$, and $I=In(v)$, since $\Phi_I=\Psi_I^J\circ \Phi_I$, by the above choice of encoding kernels we have $(F_e(X))_{e\in J}=f_e((F_{e'}(X))_{e'\in I})_{e\in J}$, i.e., the global encoding maps of the outgoing edges are functions of those of the incoming edges. For each sink node $t\in T$, $I=In(t)$, and $J=\hat{e}_t$, since $\Phi_{\hat{e}_t}=\Psi_I^J\circ \Phi_I$, the above choice of encoding kernels and decoding maps ensures that $g_t(X)=D_t((F_{e'}(X))_{e'\in I})$, i.e., each sink can obtain its demanded function value. 

Thus, a network code that satisfies all the sink demands can be obtained from a functional representation of $Q_\mathcal{E}$ satisfying (C1){$^\prime$} and (C2){$^\prime$}.
\end{proof}

Recall that matroid representation involved assignment of vectors to ground set elements. And, via the network-matroid map (proof of Theorem\ref{thm1}), edges are assigned global encoding vector. But representation of FD-relation does not pose any linearity constraint on the function assigned to edges. If the functions $\Phi_e$, $e\in \mathcal{E}$, in the representation of $Q_\mathcal{E}$ are nonlinear, then the resulting network code will also be nonlinear. Thus, FD-relations can be used to characterize network computation problems with nonlinear demands and their representation can potentially give nonlinear codes also.

\begin{example}
Consider the network computation problem given in Fig.~\ref{fig_fnc_max}. There are $11$ source nodes, labeled $1,2,\ldots,11$, each generating a $10$-bit long message. There is only one sink which wants to compute the maximum among the decimal equivalents of messages. All dashed and solid edges have same capacity (say $b$ bits). By $\max$, we mean the maximum of the decimal equivalent of the $b$-bit words. The local encoding kernels are given adjacent to the edges.
\begin{figure}
\centering
\includegraphics[scale=0.5]{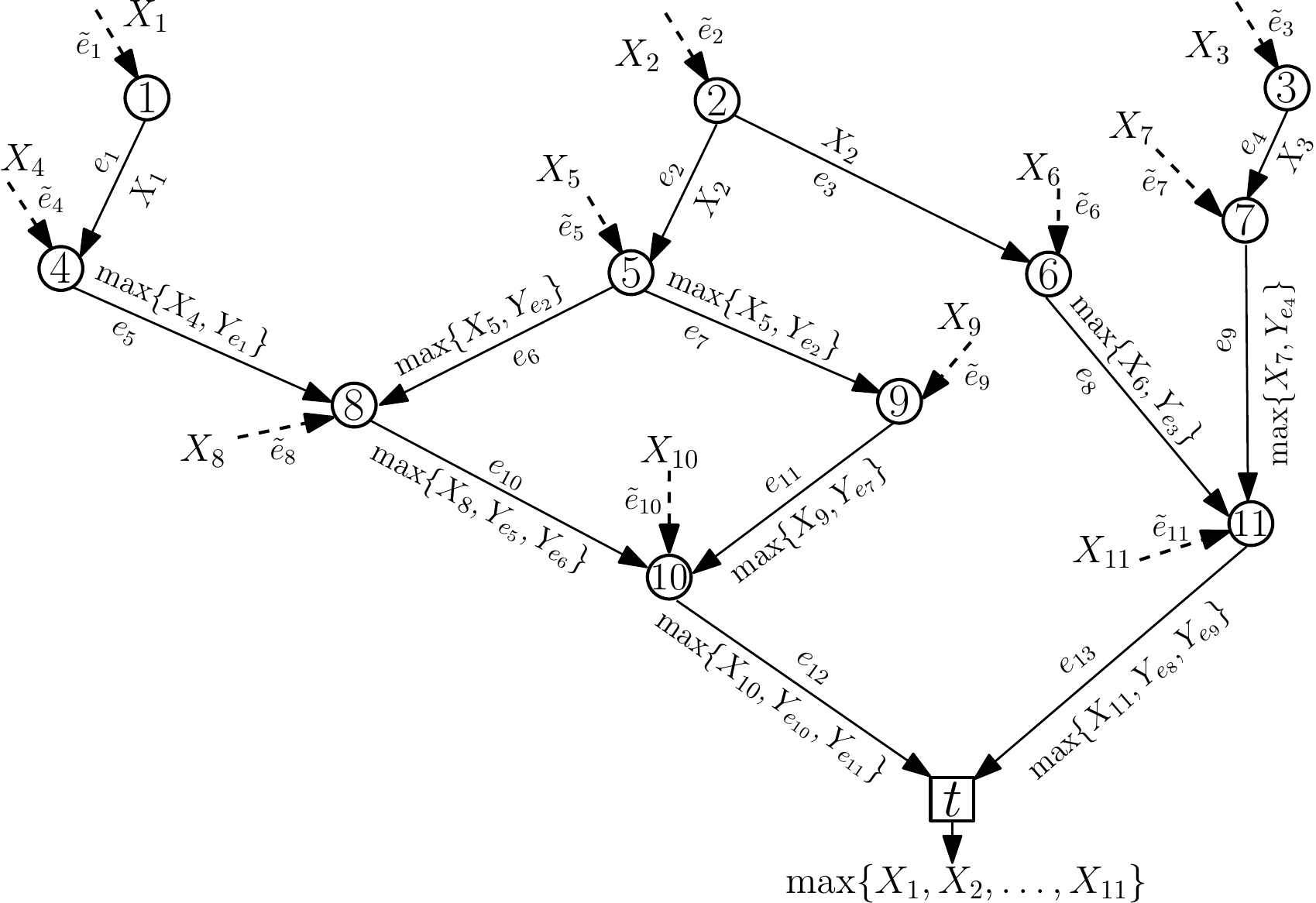}
\caption{A function computation problem.}
\label{fig_fnc_max}
\vspace{-10pt}
\end{figure}
The local and global encoding kernels of edges $e_1,e_2,\ldots,e_{13}$ are given in Table~\ref{tab_LEKGEK}. 
\begin{table}[h] 
\caption{}
\label{tab_LEKGEK}
\vspace{-6pt}
\tiny
\centering
\renewcommand{\tabcolsep}{4pt}
\renewcommand{\arraystretch}{1.5}
\begin{tabular}{|c|c|c|}
\hline
$Y_e$ & $f_e\left((Y_{e'})_{e'\in In(e)}\right)$& $F_e(X_1,\ldots,X_{11})$ \\ \hline \hline
$Y_{e_1}$ & $X_1$& $X_1$ \\ \hline
$Y_{e_2}$ & $X_2$& $X_2$ \\ \hline
$Y_{e_3}$ & $X_2$& $X_2$ \\ \hline
$Y_{e_4}$ & $X_3$& $X_3$ \\ \hline
$Y_{e_5}$ & $\max\{Y_{e_1},X_4\}$ & $\max\{X_1,X_4\}$\\ \hline
$Y_{e_6}$ & $\max\{Y_{e_2},X_5\}$ & $\max\{X_2,X_5\}$\\ \hline
$Y_{e_7}$ & $\max\{Y_{e_2},X_5\}$ & $\max\{X_2,X_5\}$\\ \hline
$Y_{e_8}$ & $\max\{Y_{e_3},X_6\}$ & $\max\{X_2,X_6\}$\\ \hline
$Y_{e_9}$ & $\max\{Y_{e_4},X_7\}$ & $\max\{X_3,X_7\}$\\ \hline
$Y_{e_{10}}$ & $\max\{Y_{e_5},Y_{e_6},X_8\}$&$\max\{X_1,X_2,X_4,X_5,X_8\}$\\ \hline
$Y_{e_{11}}$ & $\max\{Y_{e_7},X_9\}$& $\max\{X_2,X_5,X_9\}$\\ \hline
$Y_{e_{12}}$ & $\max\{Y_{e_{10}},Y_{e_{11}},X_{10}\}$& $\max\{X_1,X_2,X_4,X_5,X_8,X_9,X_{10}\}$\\ \hline
$Y_{e_{13}}$ & $\max\{Y_{e_8},Y_{e_9},X_{11}\}$& $\max\{X_2,X_3,X_6,X_7,X_{11}\}$\\ \hline
\end{tabular}
\vspace{-6pt}
\end{table}
\noindent The decoding function for sink $t$ is
\begin{align*}
D_t(Y_{e_{12}},Y_{e_{13}})\!=\!\max\{Y_{e_{12}},Y_{e_{13}}\}\!=\!\max\{X_1,X_2,\ldots,X_{11}\}.
\end{align*}

The global encoding kernels (third columns in Table~\ref{tab_LEKGEK}) form a representation (that satisfies the constraints specified in Proposition~2) for the FD-relation (defined using \eqref{eq_fdrel}) for this network computation problem. For instance, at node $8$ in Fig.~\ref{fig_fnc_max}, the sets of incoming and outgoing edges are $I=\{\tilde{e}_{8},e_5,e_6\}$ and $J=\{e_{10}\}$ respectively, and $\Phi_J(X)=F_{e_{10}}(X)=\max\{X_1,X_2,X_4,X_5,X_8\}$, $\Phi_I(X)=(F_{\tilde{e}_8}(X), F_{e_5}(X), F_{e_6}(X))=(X_8,\max\{X_1,X_4\},$ $\max\{X_2,X_5\})$, and $\Psi_I^J=f_{e_{10}}=\max$ (outputs the maximum among its $|I|=3$ input arguments). Hence, $(\Psi_I^J\!\circ\! \Phi_I)(X)=\Psi_I^J(\Phi_I(X))=\max\{F_{\tilde{e}_8}(X), F_{e_5}(X), F_{e_6}(X)\}=\max\{X_8,$ $\max\{X_1,X_4\},\max\{X_2,X_5\}\}=F_{e_{10}}(X)=\Phi_J(X)$. Similarly, at the sink node $t$, $I=\{e_{12},e_{13}\}$, $J=\{\hat{e}_t\}$, $\Phi_J(X)=\max\{X_1,X_2,\ldots,X_{11}\}$, $\Phi_I(X)=(F_{e_{12}}(X), F_{e_{13}}(X))=(\max\{X_1,X_2,X_4,X_5,X_8,X_9,X_{10}\},\max\{X_2,X_3,X_6,X_7,$ $X_{11}\})$ , and $\Psi_I^J=D_t=\max$ (outputs the maximum among its $|I|=2$ input arguments). Hence, $\Psi_I^J(\Phi_I(X))=\max\{F_{e_{12}}(X),F_{e_{13}}(X)\}=F_{\hat{e}_t}(X)=\Phi_J(X)$. \hfill $\square$
\end{example}

\section{Discussion}\label{sec_disc}
In this paper, we established a relationship between network computation problems and representable matroids. The definition of matroidal networks given in \cite{Dough} for network coding problems was modified to fit the requirements of network computation problems; the modifiend definition subsumes the original definition of \cite{Dough} as a special case. We proved that a network computation problem with linear functions demanded by sinks admits a scalar linear solution if and only if it is matroidal with respect to a representable matroid whose representation satisfies certain constraints imposed by the network computation problem. An implication of the proposed work is that results from theory of matroid representability can be used for network computation problems. Lastly, we explored relation between FD-relations and network computation problem and concluded that FD-relations can characterize nonlinear codes also.


\section*{Acknowledgment}
This work was supported partly by the Science and Engineering Research Board (SERB) of Department of Science and Technology (DST), Government of India, through J.C. Bose National Fellowship to B. Sundar Rajan.




\begin{thebibliography}{10} 
\bibitem{Ahls}
R. Ahlswede, N. Cai, S.-Y.~R. Li, and R.~W. Yeung, ``Network Information Flow,''
\textit{IEEE Trans. Inf. Theory}, vol. 46, no. 4, pp. 1204-1216, July 2000.

\bibitem{Yeung}
R.~W. Yeung, \textit{Information Theory and Network Coding}. New York, NY, USA: Springer, 2008.



\bibitem{Ho}
T. Ho and D.~S. Lun, \textit{Network Coding An Introduction}. New York, NY, USA: Cambridge University Press, 2008.
%


\bibitem{Appu}
R. Appuswamy, M. Franceschetti, N. Karamchandani, and K. Zeger, ``Network Coding for Computing: Cut-set Bounds,'' \textit{IEEE Trans. Inf. Theory}, vol. 57, no. 2, pp. 1015-1030, February 2011.

\bibitem{GiriKr}
A. Giridhar and P.~R. Kumar, ``Computing and Communicating Functions Over Sensor Networks,'' \textit{IEEE J. Sel. Areas Commun.}, vol. 23, no. 4, pp. 755-764, April 2005.

\bibitem{Dey}
B.~K. Rai and B.~K. Dey, ``On Network Coding for Sum-Networks,'' \textit{IEEE Trans. Inf. Theory}, vol. 58, no. 1, pp. 50-63, January 2012.

\bibitem{Dey2}
V. Shah, B.~K. Dey, and D. Manjunath, ``Network Flows for Function Computation,'' \textit{IEEE J. Sel. Areas Commun.}, vol. 31, no. 4, pp. 714-730, April 2013.

\bibitem{Appu2}
R. Appuswamy, M. Franceschetti, N. Karamchandani, and K. Zeger, ``Linear Codes, Target Function Classes, and Network Computing Capacity,'' \textit{IEEE Trans. Inf. Theory}, vol. 59, no. 9, pp. 5741-5753, September 2013.

%

%


%
%
%
%
%
%


\bibitem{Dough}
R. Dougherty, C. Freiling, and K. Zeger, ``Networks, Matroids, and Non Shannon Information Inequalities,'' \textit{IEEE Trans. Inf. Theory}, vol. 53, no. 6, pp. 1949-1969, June 2007.

\bibitem{KimMed}
A. Kim and M. M\'{e}dard, ``Scalar-linear Solvability of Matroidal Networks Associated with Representable Matroids,'' in \textit{Proc. 6th Int. Symp. Turbo Codes \& Iterative Information Processing}, Brest, France, 2010, pp. 452-456. 


\bibitem{ICRel}
S. El~Rouayheb, A. Sprintson, and C. Georghiades, ``On the Index Coding Problem and its Relation to Network Coding and Matroid Theory,'' \textit{IEEE Trans. Inf. Theory}, vol. 56, no. 7, pp. 3187-3195, July 2010.

\bibitem{ICRel2}
M. Effros, S. El~Rouayheb, and M. Langberg, ``An Equivalence Between Network Coding and Index Coding,'' \textit{IEEE Trans. Inf. Theory}, vol. 61, no. 5, pp. 2478-2487, May 2015.

\bibitem{FICP}
A. Gupta and B.~S. Rajan, ``Error-Correcting Functional Index Codes, Generalized Exclusive Laws and Graph Coloring,'' \textit{IEEE Int. Conf. Commun.}, Kuala Lumpur, Malaysia, 2016, [Online]. Available: http://arxiv.org/abs/1510.04820.

\bibitem{FNC-FIC}
A. Gupta and B.~S. Rajan, ``A Relation Between Network Computation and Functional Index Coding Problems,'' Accepted for presentation at \textit{IEEE Information Theory Workshop}, Cambridge, UK, 2016, [Online]. Available: http://arxiv.org/abs/1603.05365.

\bibitem{FDNC}
S. El~Rouayheb, A. Sprintson, and C. Georghiades, ``A New Construction Method for Networks from Matroids,'' in \textit{Proc. IEEE Int. Symp. Information Theory}, Seoul, South Korea, 2009, pp. 2872-2876. 

\bibitem{Oxley}
J.~G. Oxley, \textit{Matroid Theory}. Oxford, UK: Oxford University Press, 1992.

\bibitem{Welsh}
D.~J.~A. Welsh, \textit{Matroid Theory}. Mineola, NY, USA: Dover Publications, Inc., 2010.

\bibitem{FD}
F. Mat\'{u}\v{s}, ``Abstract Functional Dependency Structures,'' \textit{Theoretical Comput. Sci.}, vol. 81, no. 1, pp. 117-126, April 1991.

\end{thebibliography}
\end{document}